\documentclass[11pt]{article}

\usepackage{a4}
\usepackage[top=30truemm,bottom=30truemm]{geometry}

\usepackage{authblk}
\usepackage{graphicx}
\usepackage{rotating}
\usepackage{comment}
\usepackage{multirow}
\usepackage{amsmath}
\usepackage{amssymb}
\usepackage{amsthm}
\usepackage{algorithm}
\usepackage[noend]{algpseudocode}
\usepackage{algorithmicx}
\usepackage{url}
\usepackage{enumerate}
\usepackage{bm}
\usepackage{setspace}
\usepackage{epsfig}
\usepackage{color}

\usepackage[draft,mode=multiuser]{fixme}

\usepackage[whole]{bxcjkjatype}

\fxuseenvlayout{colorsig}
\fxusetargetlayout{color}
\FXRegisterAuthor{si}{asi}{SI}
\fxsetface{env}{}

\newcommand{\Suffix}{\mathsf{Suffix}}
\newcommand{\STree}{\mathsf{STree}}
\newcommand{\DAWG}{\mathsf{DAWG}}

\newcommand{\polylog}{\mathrm{polylog}}

\newcommand{\suflink}{\mathsf{slink}}

\newcommand{\Long}{\mathit{long}}
\newcommand{\rev}[1]{\overline{#1}}

\def\W#1#2{{\mathsf{W\_link}}_{#1}({#2})}

\def\fn(#1){\fnbf{#1}}

\newtheorem{theorem}{Theorem}
\newtheorem{lemma}{Lemma}

\title{
  Pointer-Machine Algorithms for Fully-Online Construction of \\
  Suffix Trees and DAWGs on Multiple Strings
}

\author{Shunsuke Inenaga}

\affil{\textit{Department of Informatics, Kyushu University, Japan} \\
    \textit{PRESTO, Japan Science and Technology Agency, Japan} \\
\texttt{\small inenaga@inf.kyushu-u.ac.jp}
}

\date{}

\begin{document}
\maketitle

\begin{abstract}
  We deal with the problem of maintaining
  the {\em suffix tree} indexing structure
  for a fully-online collection of multiple strings,
  where a new character can be prepended
  to any string in the collection at any time.
  The only previously known algorithm for the problem,
  recently proposed by Takagi et al. [Algorithmica 82(5): 1346-1377 (2020)],
  runs in $O(N \log \sigma)$ time and $O(N)$ space \emph{on the word RAM} model,
  where $N$ denotes the total length of the strings
  and $\sigma$ denotes the alphabet size.
  Their algorithm makes heavy use of the nearest marked ancestor (NMA)
  data structure on semi-dynamic trees,
  that can answer queries and supports insertion of nodes
  in $O(1)$ amortized time on the word RAM model.
  In this paper, we present a simpler fully-online right-to-left
  algorithm that builds the suffix tree for a given string collection
  in $O(N (\log \sigma + \log d))$ time and $O(N)$ space,
  where $d$ is the maximum number of in-coming \emph{Weiner links} to a node
  of the suffix tree.
  We note that $d$ is bounded by the height of the suffix tree,
  which is further bounded by the length of the longest string in the collection.
  The advantage of this new algorithm is that it works on
  \emph{the pointer machine model}, namely, it does not use the complicated
  NMA data structures that involve table look-ups.
  As a byproduct, we also obtain a pointer-machine algorithm
  for building the \emph{directed acyclic word graph} (\emph{DAWG})
  for a fully-online left-to-right collection of multiple strings,
  which runs in $O(N (\log \sigma + \log d))$ time and $O(N)$ space
  again without the aid of the NMA data structures.
\end{abstract}

%% BODY %%%%%%%%%%%%%%%%%%%%%%%%% 

\section{Introduction}  

\subsection{Suffix trees and DAWGs}

\emph{Suffix trees} are a fundamental string data structure
with a myriad of applications~\cite{GusfieldS04}.
The first efficient construction algorithm for suffix trees,
proposed by Weiner~\cite{Weiner}, builds the suffix tree for a string
in a right-to-left online manner,
by updating the suffix tree each time a new character is prepended to the string.
It runs in $O(n \log \sigma)$ time and $O(n)$ space,
where $n$ is the length of the string and $\sigma$ is the alphabet size.

One of the most interesting features of Weiner's algorithm
is a very close relationship to
Blumer et al.'s algorithm~\cite{blumer85:_small_autom_recog_subwor_text}
that builds the \emph{directed acyclic word graph} (\emph{DAWG})
in a left-to-right online manner,
by updating the DAWG each time a new character is prepended to the string.
It is well known that the DAG of the Weiner links of the suffix tree of $T$ is equivalent to the DAWG of the reversal $rev{T}$ of $T$, or symmetrically,
the suffix link tree of the DAWG of $\rev{T}$ is equivalent to
the suffix tree of $T$.
Thus, right-to-left online construction of suffix trees
is essentially equivalent to left-to-right construction of DAWGs.
This means that Blumer et al.'s DAWG construction algorithm also runs
in $O(n \log \sigma)$ time and $O(n)$ space~\cite{blumer85:_small_autom_recog_subwor_text}.

DAWGs also support efficient pattern matching queries,
and have been applied to other important string problems such as
local alignment~\cite{DoS13},
pattern matching with variable-length don't cares~\cite{KucherovR97},
dynamic dictionary matching~\cite{HendrianIYS19},
compact online Lempel-Ziv factorization~\cite{YamamotoIBIT14},
finding minimal absent words~\cite{FujishigeTIBT16},
and finding gapped repeats~\cite{TanimuraFIIBT15},
on the input string.

\subsection{Fully online construction of suffix trees and DAWGs}

Takagi et al.~\cite{TakagiIABH20} initiated
the generalized problem of maintaining
the suffix tree for a collection of strings in a \emph{fully-online manner},
where a new character can be prepended to any string in the collection
at any time.
This fully-online scenario arises in 
real-time database systems e.g. for sensor networks or trajectories.
Takagi et al. showed that a direct application of Weiner's algorithm~\cite{Weiner}
to this fully-online setting requires to visit $\Theta(N\min(K, \sqrt N))$ nodes,
where $N$ is the total length of the strings
and $K$ is the number of strings in the collection.
Note that this leads to a worst-case 
$\Theta(N^{1.5} \log \sigma)$-time construction when $K = \Omega(\sqrt{N})$.

In their analysis, it was shown that Weiner's original algorithm
applied to a fully-online string collection
visits a total of $\Theta(N\min(K, \sqrt N))$ nodes.
This means that the amortization argument of Weiner's algorithm
for the number of nodes visited in the climbing process for inserting
a new leaf, does not work for multiple strings in the fully-online setting.
To overcome difficulty, Takagi et al. proved the three following arguments:
(1) By using $\sigma$ nearest marked ancestor (NMA) structures~\cite{westbrook92:_fast_increm_planar_testin},
one can skip the last part of the climbing process;
(2) All the $\sigma$ NMA data structures can be stored in $O(n)$ space;
(3) The number of nodes explicitly visited in the remaining part of each
climbing process can be amortized $O(1)$ per new added character.
This led to their $O(N \log \sigma)$-time and $O(N)$-space fully-online
right-to-left construction of the suffix tree for multiple strings.

Takagi et al.~\cite{TakagiIABH20}
also showed that Blumer et al.'s algorithm applied to 
a fully-online left-to-right DAWG construction 
requires at least $\Theta(N\min(K, \sqrt N))$ work as well.
%Here, by an explicit representation we mean the DAWG where edges are
%implemented by pointers.
They also showed how to maintain an \emph{implicit} representation
of the DAWG of $O(N)$ space which supports fully-online updates
and simulates a DAWG edge traversal in $O(\log \sigma)$ time each.
The key here was again the non-trivial use of
the aforementioned $\sigma$ NMA data structures
over the suffix tree of the reversed strings.

As was stated above,
Takagi et al.'s construction heavily relies on the use of the NMA data structures~\cite{westbrook92:_fast_increm_planar_testin}.
Albeit NMA data structures are useful and powerful,
all known NMA data structures for (static and dynamic) trees
that support $O(1)$ (amortized) time queries and updates~\cite{gabow85:_linear_time_algor_special_case,ImaiA87,westbrook92:_fast_increm_planar_testin} are quite involved,
and they are valid only on the word RAM model as they use look-up tables
that explicitly store the answers for small sub-problems.
Hence, in general, it would be preferable if one can achieve
similar efficiency without NMA data structures.

\subsection{Our contribution}

In this paper, we show how to maintain
the suffix tree for a right-to-left fully-online string collection
in $O(N (\log \sigma + \log d))$ time and $O(N)$ space,
where $d$ is the maximum number of in-coming \emph{Weiner links} to a node of the suffix tree.
Our construction does not use NMA data structures
and works in the \emph{pointer-machine model}~\cite{Tarjan79},
which is a simple computational model without address arithmetics.
We note that $d$ is bounded by the height of the suffix tree.
Clearly, the height of the suffix tree is at most the maximum length
of the strings.
Hence, the $d$ term can be dominated by the $\sigma$ term
when the strings are over integer alphabets of polynomial size in $N$,
or when a large number of strings of similar lengths are treated.
To achieve the aforementioned bounds on the pointer-machine model, we reduce the problem of maintaining
in-coming Weiner links of nodes to the \emph{ordered split-insert-find problem},
which maintains dynamic sets of sorted elements allowing for
split and insert operations, and find queries,
which can be solved in a total of $O(N \log d)$ time and $O(N)$ space.

As a byproduct of the above result, we also obtain
the \emph{first} non-trivial algorithm that
maintains an \emph{explicit} representation of the DAWG for
fully-online left-to-right multiple strings,
which runs in $O(N (\log \sigma + \log d))$ time and $O(N)$ space.
By an explicit representation, we mean that every edge
of the DAWG is implemented as a pointer.
This DAWG construction does not require complicated table look-ups
and thus also works on the pointer machine model.

\section{Preliminaries}

\subsection{String notations}

Let $\Sigma$ be a general ordered alphabet.
Any element of $\Sigma^*$ is called a \emph{string}.
For any string $T$, let $|T|$ denote its length.
Let $\varepsilon$ be the empty string, namely, $|\varepsilon| = 0$.
Let $\Sigma^+ = \Sigma \setminus \{\varepsilon\}$.
If $T = XYZ$, then $X$, $Y$, and $Z$ are called 
a \emph{prefix}, a \emph{substring}, and a \emph{suffix} of $T$, respectively.
For any $1 \leq i \leq j \leq |T|$,
let $T[i..j]$ denote the substring of $T$ that begins at position $i$
and ends at position $j$ in $T$.
%For convenience, let $T[i..j] = \varepsilon$ if $i > j$.
For any $1 \leq i \leq |T|$, let $T[i]$ denote the $i$th character of $T$.
For any string $T$, let $\Suffix(T)$ denote the set of suffixes of $T$,
and for any set $\mathcal{T}$ of strings,
let $\Suffix(\mathcal{T})$ denote the set of suffixes of all strings in $\mathcal{T}$.
Namely, $\Suffix(\mathcal{T}) = \bigcup_{T \in \mathcal{T}} \Suffix(T)$.
For any string $T$, let $\rev{T}$ denote the reversed string of $T$,
i.e., $\rev{T} = T[|T|] \cdots T[1]$.
For any set $\mathcal{T}$ of strings,
let $\rev{\mathcal{T}} = \{\rev{T} \mid T \in \mathcal{T}\}$.

\subsection{Suffix trees and DAWGs for multiple strings}

For ease of description,
we assume that each string $T_i$ in the collection $\mathcal{T}$
terminates with a unique character $\$_i$ that does not appear
elsewhere in $\mathcal{T}$.
However, our algorithms work without $\$_i$ symbols at the right end of strings
as well.

A compacted trie is a rooted tree such that 
(1) each edge is labeled by a non-empty string,
(2) each internal node is branching, and
(3) the string labels of the out-going edges of each node begin with mutually distinct characters.
The \emph{suffix tree}~\cite{Weiner} 
for a text collection $\mathcal{T}$, denoted $\STree(\mathcal{T})$,
is a compacted trie which represents $\Suffix(\mathcal{T})$.
The \emph{string depth} of a node $v$ of $\Suffix(\mathcal{T})$
is the length of the substring that is represented by $v$.
We sometimes identify node $v$ with the substring it represents.
%For any non-root node $v$ of $\STree(\mathcal{T})$,
%let $\parent(v)$ denote the parent of $v$.
The suffix tree for a single string $T$ is denoted $\STree(T)$.

\begin{figure}[t]
  \begin{center}
    \includegraphics[scale=0.6]{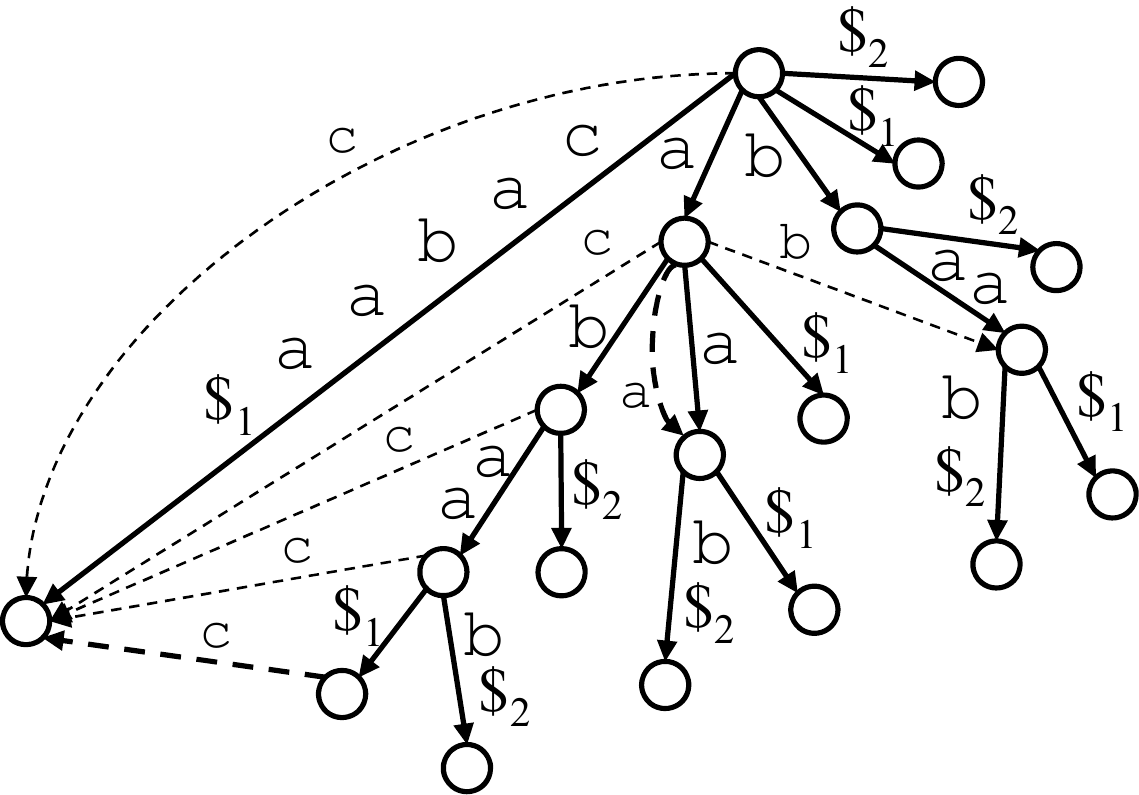}
    \hfill
    \includegraphics[scale=0.6]{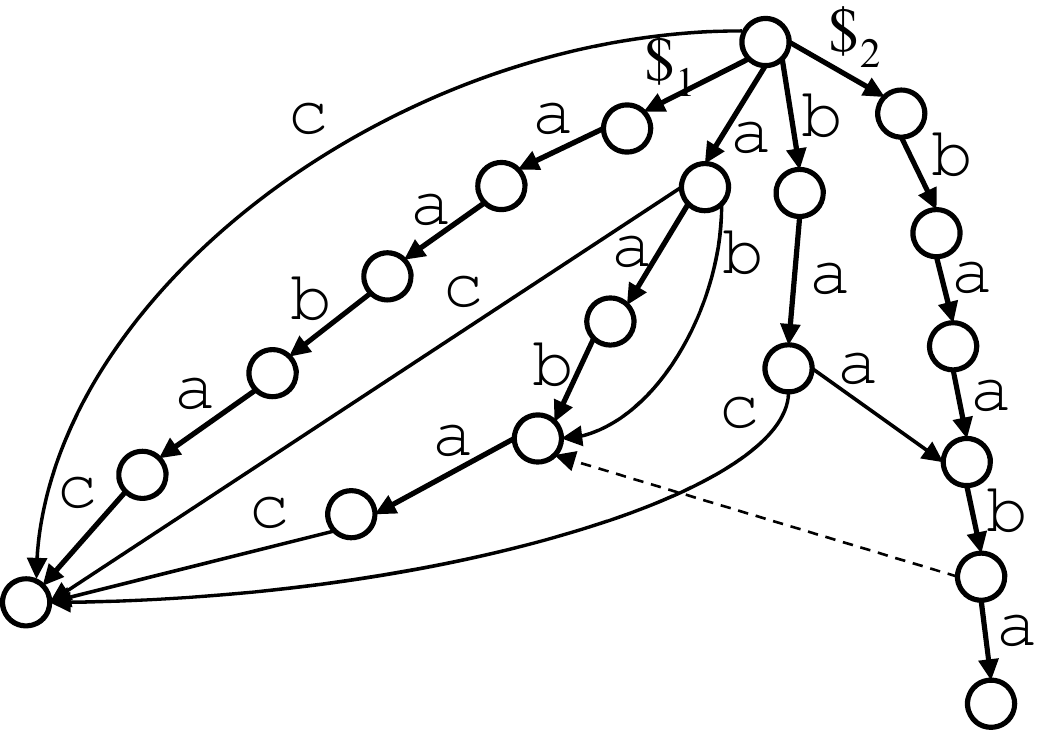}
  \end{center}
  \caption{Left: $\STree(\mathcal{T})$ for $\mathcal{T} = \{\mathtt{cabaa\$_1}, \mathtt{abaab\$_2}\}$. The bold broken arrows represent hard Weiner links, while the narrow broken arrows represent soft Weiner links. Not all Weiner links are shown for simplicity. Right: $\DAWG(\mathcal{S})$ for $\mathcal{S} = \rev{\mathcal{T}} = \{\mathtt{\$_1aabac}, \mathtt{\$_2baaba}\}$. The broken arrow represents a suffix link. Not all suffix links are shown for simplicity.}
  \label{fig:ST_and_DAWG}
\end{figure}

$\STree(\mathcal{T})$ has at most $2N-1$ nodes and thus $2N-2$ nodes,
since every internal node of $\STree(\mathcal{T})$ is branching
and there are $N$ leaves in $\STree(\mathcal{T})$.
By representing each edge label $x$ with a triple $\langle k, i, j \rangle$
of integers such that $x = T_k[i..j]$,
$\STree(\mathcal{T})$ can be stored with $O(N)$ space.

We define the \emph{suffix link} of each non-root node $av$
of $\STree(\mathcal{T})$ with $a \in \Sigma$ and $v \in \Sigma^*$,
by $\suflink(av) = v$.
For each explicit node $v$ and $a \in \Sigma$,
we also define the reversed suffix link (a.k.a. Weiner link)
by $\W{a}{v} = avx$, where $x \in \Sigma^*$ is the shortest string
such that $avx$ is a node of $\STree(\mathcal{T})$.
$\W{a}{v}$ is undefined if $av$ is not a substring of strings in $\mathcal{T}$.
A Weiner link $\W{a}{v} = avx$ is said to be \emph{hard}
if $x = \varepsilon$, and \emph{soft} if $x \in \Sigma^+$.
%Let $\sw$ be a Boolean function such that 
%for any explicit node $v$ and $a \in \Sigma$,
%$\sw_a(v) = 1$ iff (soft or hard) W-link $\W{a}{v}$ exists.
%Notice that if $\sw_a(v) = 1$ for a node $v$ and $a \in \Sigma$,
%then $\sw_a(u) = 1$ for every ancestor $u$ of $v$.

See the left diagram of Figure~\ref{fig:ST_and_DAWG} for
an example of $\STree(\mathcal{T})$ and Weiner links.

The \emph{directed acyclic word graph} (\emph{DAWG} in short)~\cite{blumer85:_small_autom_recog_subwor_text,Blumer87}
of a text collection $\mathcal{S}$, denoted $\DAWG(\mathcal{S})$,
is a (partial) DFA which represents $\Suffix(\mathcal{S})$.
It is proven in~\cite{Blumer87} that
$\DAWG(\mathcal{S})$ has at most $2N-1$ nodes
and $3N-4$ edges for $N \geq 3$.
Since each DAWG edge is labeled by a single character,
$\DAWG(\mathcal{S})$ can be stored with $O(N)$ space.
The DAWG for a single string $S$ is denoted $\DAWG(S)$.

A node of $\DAWG(\mathcal{S})$ corresponds to the
substrings in $\mathcal{S}$ which
share the same set of ending positions in $\mathcal{S}$.
Thus, for each node, there is a unique longest string
represented by that node.
For any node $v$ of $\DAWG(\mathcal{S})$,
let $\Long(v)$ denote the longest string represented by $v$.
%For simplicity, we identify each node $v$ with its longest string $\Long(v)$.
An edge $(u, a, v)$ in the DAWG is called 
\emph{primary} if $|\Long(u)| + 1 = |\Long(v)|$,
and is called \emph{secondary} otherwise.
For each node $v$ of $\DAWG(\mathcal{S})$ with $|\Long(v)| \geq 1$,
let $\suflink(v) = y$, where
$y$ is the longest suffix of $\Long(v)$ which is not represented by $v$.

Suppose $S = \rev{\mathcal{T}}$.
It is known (c.f.~\cite{blumer85:_small_autom_recog_subwor_text,Blumer87,cr:94})
that there is a node $v$ in $\STree(\mathcal{T})$ iff
there is a node $x$ in $\DAWG(\mathcal{S})$
such that $\Long(x) = \rev{v}$.
Also, the hard Weiner links and the soft Weiner links
of $\STree(\mathcal{T})$ coincide with
the primary edges and the secondary edges of $\DAWG(\mathcal{S})$, respectively.
In a symmetric view, the reversed suffix links of
$\DAWG(\mathcal{S})$ coincide with the suffix tree
$\STree(\mathcal{T})$ for $\mathcal{T}$.

See Figure~\ref{fig:ST_and_DAWG} for some concrete examples of the aforementioned symmetry.
For instance, the nodes $\mathtt{abaa}$ and $\mathtt{baa}$
of $\STree(\mathcal{T})$ correspond to
the nodes of $\DAWG(\mathcal{S})$ whose longest strings are
$\rev{\mathtt{abaa}} = \mathtt{aaba}$ and
$\rev{\mathtt{baa}} = \mathtt{aab}$, respectively.
Observe that both $\STree(\mathcal{T})$ and $\DAWG(\mathcal{S})$ have
19 nodes each.
The Weiner links of $\STree(\mathcal{T})$ labeled by character $\mathtt{c}$
correspond to the out-going edges of $\DAWG(\mathcal{S})$ labeled by $\mathtt{c}$.
To see another example, the three Weiner links from node $\mathtt{a}$
in $\STree(\mathcal{T})$ labeled $\mathtt{a}$, $\mathtt{b}$, and $\mathtt{c}$
correspond to the three out-going edges of node $\{\mathtt{a}\}$
of $\DAWG(\mathcal{S})$
labeled $\mathtt{a}$, $\mathtt{b}$, and $\mathtt{c}$, respectively.
For the symmetric view, focus on the suffix link of the node
$\{\mathtt{\$_2 baab}, \mathtt{baab}\}$ of $\DAWG(\mathcal{S})$
to the node $\{\mathtt{aab}, \mathtt{ab}\}$.
This suffix link reversed corresponds to the edge labeled $\mathtt{b\$_2}$
from the node $\mathtt{baa}$ to the node $\mathtt{baab\$_2}$
in $\STree(\mathcal{T})$.

We now see that the two following tasks are essentially equivalent:
\begin{itemize}
\item[(A)] Building $\STree(\mathcal{T})$ for a fully-online
right-to-left text collection $\mathcal{T}$, using hard and soft Weiner links.
\item[(B)] Building $\DAWG(\mathcal{S})$ for a fully-online
left-to-right text collection $\mathcal{S}$, using suffix links.
\end{itemize}

\subsection{Pointer machines}

\emph{A pointer machine}~\cite{Tarjan79} is an abstract model of computation
such that the state of computation is stored as a directed graph,
where each node can contain a constant number of data (e.g. integers, symbols)
and a constant number of pointers (i.e. out-going edges to other nodes).
The instructions supported by the pointer machine model are basically
creating new nodes and pointers, manipulating data, and performing comparisons.
The crucial restriction in the pointer machine model,
which distinguishes it from the word RAM model,
is that pointer machines cannot perform address arithmetics,
namely, memory access must be performed only by
an explicit reference to a pointer.
While the pointer machine model is apparently weaker
than the word RAM model that supports address arithmetics and 
unit-cost bit-wise operations,
the pointer machine model serves as a good basis for
modeling linked structures such as trees and graphs,
which are exactly our targets in this paper.
In addition, pointer-machines are powerful enough to simulate
list-processing based languages such as LISP and Prolog (and their variants),
which have recurrently gained attention.

\section{Brief reviews on previous algorithms}

To understand why and how our new algorithms to be presented in
Section~\ref{sec:algorithms} work efficiently,
let us briefly recall the previous related algorithms. 

\subsection{Weiner's algorithm and Blumer et al.'s algorithm for a single string}
\label{sec:Weiner_Blumer}

First, we briefly review how Weiner's algorithm for a single string $T$
adds a new leaf to the suffix tree when a new character $a$
is prepended to $T$.
Our description of Weiner's algorithm slightly differs from the original one,
in that we use both hard and soft Weiner links
while Weiner's original algorithm uses hard Weiner links only
and it instead maintains Boolean vectors indicating the existence of soft Weiner links.

Suppose we have already constructed $\STree(T)$ with hard and soft Weiner links.
Let $\ell$ be the leaf that represents $T$.
Given a new character $a$, Weiner's algorithm climbs up the path
from the leaf $\ell$ until encountering the deepest ancestor $v$ of $\ell$
that has a Weiner link $\W{a}{v}$ defined.
If there is no such ancestor of $\ell$ above,
then a new leaf representing $aT$ is inserted from the root $r$ of the suffix tree.
Otherwise, the algorithm follows the Weiner link $\W{a}{v}$ and 
arrives at its target node $u = \W{a}{v}$.
There are two sub-cases:
\begin{itemize}
  \item[(1)] If $\W{a}{v}$ is a hard Weiner link,
then a new leaf $\hat{\ell}$ representing $aT$ is inserted from $u$.

  \item[(2)] If $\W{a}{v}$ is a soft Weiner link,
then the algorithm splits the incoming edge of $u$ into two edges
by inserting a new node $y$ as a new parent of $u$
such that $|y| = |v|+1$
(See also Figure~\ref{fig:Weiner_case_2}). 
A new leaf representing $aT$ is inserted from this new internal node $y$.
We also copy each \emph{out-going} Weiner link $\W{c}{u}$ from $u$
with a character $c$ as an out-going Weiner link $\W{c}{y}$ from $y$
so that their target nodes are the same (i.e. $\W{c}{u}$ = $\W{c}{y}$).
See also Figure~\ref{fig:Weiner_case_2_copy}.
Then, a new \emph{hard} Weiner link is created from $v$ to $y$ with label $a$,
in other words, an old soft Weiner link $\W{a}{v} = u$
is \emph{redirected} to a new hard Weiner link $\W{a}{v} = y$.
In addition, all the old soft Weiner links of ancestors $z$ of $v$
such that $\W{a}{z} = u$ in $\STree(T)$
have to be redirected to new soft Weiner links $\W{a}{z} = y$
in $\STree(aT)$.
These redirections can be done by keeping climbing up the path from
$v$ until finding the deepest node $x$ that has a hard Weiner link
with character $a$ pointing to the parent of $u$ in $\STree(T)$.
\end{itemize}

\begin{figure}[t]
  \begin{center}
    \includegraphics[scale=0.8]{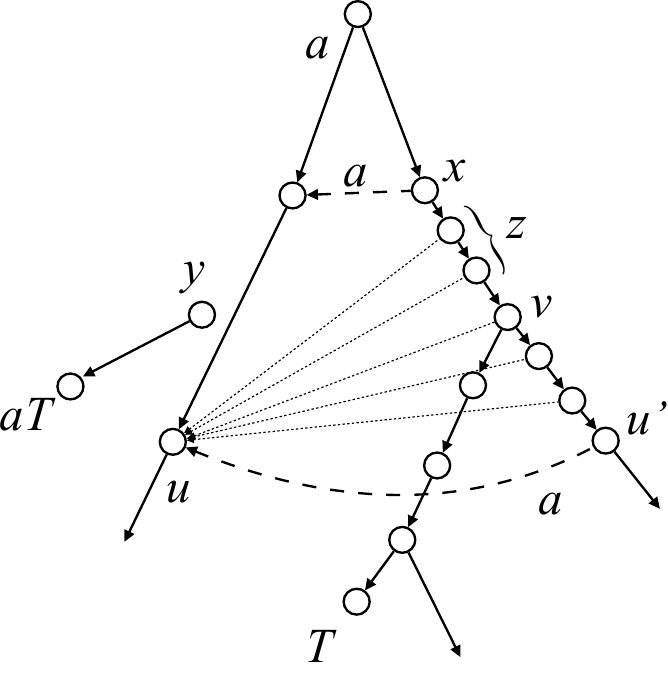}
    \hfil
    \includegraphics[scale=0.8]{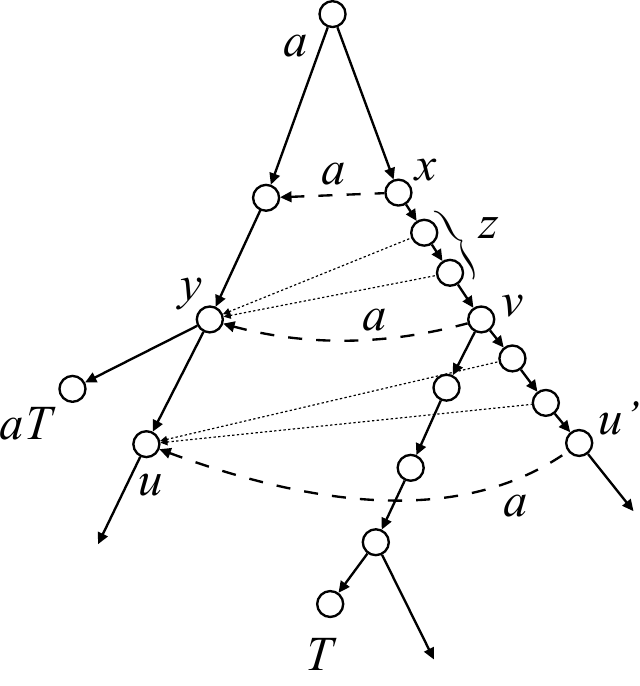}
  \end{center}
  \caption{Left: Illustration for $\STree(T)$ of Case (2) before inserting the new leaf representing $aT$. Right: Illustration for $\STree(aT)$ of Case (2) after inserting the new leaf representing $aT$. In both diagrams, thick broken arrows represent hard Winer links, and narrow broken arrows represent soft Weiner links. All these Winer links are labeled by $a$. Also, new Weiner links labeled $a$ are created from the nodes between the leaf for $T$ and $v$ to the new leaf for $aT$ (not shown in this diagram).}
  \label{fig:Weiner_case_2}
\end{figure}

\begin{figure}[t]
  \begin{center}
    \includegraphics[scale=0.85]{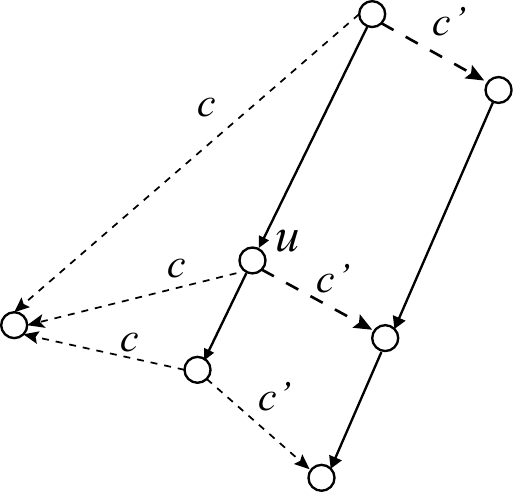}
    \hfil
    \includegraphics[scale=0.85]{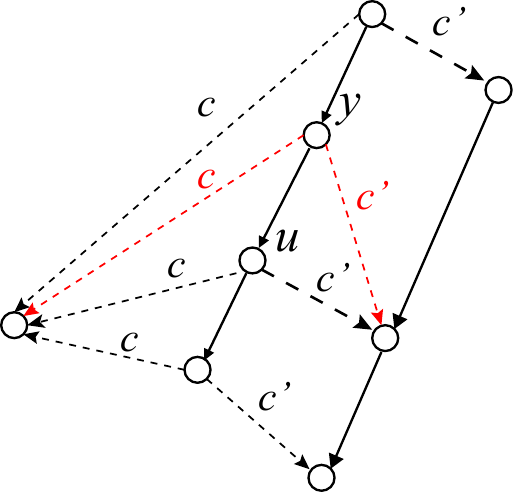}
  \end{center}
  \caption{Illustration of the copy process of the out-going Weiner links of $u$ to its new parent $y$ in Case (2). Left: Out-going Weiner links of node $u$ before the update. Right: Each out-going Winer link of node $u$ is copied to its new parent $y$, represented by a red broken arrow.}
  \label{fig:Weiner_case_2_copy}
\end{figure}

In both Cases (1) and (2) above,
new soft Weiner links $\W{a}{x} = \hat{\ell}$ are created
from every node $x$ in the path from $\ell$ to the child of $v$.

The running time analysis of the above algorithm has three phases.
\begin{itemize}
  \item[(a)] In both Cases (1) and (2),
the number of nodes from leaf $\ell$ for $T$ to $v$
is bounded by the number of newly created soft Weiner links.
This is amortized $O(1)$ per new character since
the resulting suffix tree has a total of $O(n)$ soft Weiner links~\cite{blumer85:_small_autom_recog_subwor_text},
where $n = |T|$.

  \item[(b)] In Case (2), the number of out-going Weiner links copied from $u$ to $y$
is bounded by the number of newly created Weiner link,
which is also amortized $O(1)$ per new character by the same argument as (a).

  \item[(c)] In Case (2), the number redirected soft Weiner links
is bounded by the number of nodes from $v$ to $x$.
The analysis by Weiner~\cite{Weiner} shows that
this number of nodes from $v$ to $x$ can be amortized $O(1)$.
\end{itemize}
Wrapping up (a), (b), and (c),
the total numbers of visited nodes, created Weiner links,
and redirected Weiner links
through constructing $\STree(T)$ by prepending $n$ characters are $O(n)$.
Thus Weiner's algorithm constructs $\STree(T)$
in right-to-left online manner in $O(n \log \sigma)$ time with $O(n)$ space,
where the $\log \sigma$ term comes from the cost for
maintaining Weiner links of each node in the lexicographically sorted order
by e.g. a standard balanced binary search tree.

Since this algorithm correctly maintains all (hard and soft) Weiner links,
it builds $\DAWG(S)$ for the reversed string $S = \rev{T}$
in a left-to-right manner, in $O(n \log \sigma)$ time with $O(n)$ space.
In other words, this version of Weiner's algorithm is equivalent to
Blumer et al.'s DAWG online construction algorithm.

We remark that the aforementioned version of Weiner's algorithm,
and equivalently Blumer et al.'s algorithm, work on the pointer machine model
as they do not use address arithmetics nor table look-ups.

\subsection{Takagi et al.'s algorithm for multiple strings on the word RAM}

When Weiner's algorithm is applied to
fully-online right-to-left construction of $\STree(\mathcal{T})$,
the amortization in Analysis (c) does not work.
Namely, it was shown by Takagi et al.~\cite{TakagiIABH20}
that the number of redirected
soft Weiner links is $\Theta(N\min(K, \sqrt N))$
in the fully-online setting for multiple $K$ strings.
A simpler upper bound $O(NK)$ immediately follows from an observation that
the insertion of a new leaf for a string $T_i$ in $\mathcal{T}$
may also increase the depths of the leaves for all the other $K-1$
strings $T_1, \ldots, T_{i-1}, T_{i+1}, \ldots, T_{K}$ in $\mathcal{T}$.
Takagi et al. then obtained the aforementioned improved $O(N\min(K, \sqrt N))$ upper bound,
and presented a lower bound instance that indeed requires $\Omega(N\min(K, \sqrt N))$ work.
It should also be noted that the original version of Weiner's algorithm that only maintains
Boolean indicators for the existence of soft Weiner links,
must also visit $\Theta(N\min(K, \sqrt N))$ nodes~\cite{TakagiIABH20}.

Takagi et al. gave a neat way to overcome this difficulty
by using the nearest marked ancestor (NMA) data structure~\cite{westbrook92:_fast_increm_planar_testin} for a rooted tree.
This NMA data structure allows for
making unmarked nodes, splitting edges, inserting new leaves,
and answering NMA queries in $O(1)$ amortized time each,
in the word RAM model of machine word size $\Omega(\log N)$.
Takagi et al. showed how to skip the nodes between $v$ to $x$ in $O(1)$ amortized time
using a single NMA query on the NMA data structure associated to
a given character $a$ that is prepended to $T$.
They also showed how to store $\sigma$ NMA data structures for 
all $\sigma$ distinct characters in $O(N)$ total space.
Since the amortization argument (c) is no more needed
by the use of the NMA data structures,
and since the analyses (a) and (b) still hold for fully-online multiple strings,
the total number of visited nodes was reduced to $O(N)$
in their algorithm.
This led their construction in $O(N \log \sigma)$ time and $O(N)$ space,
in the word RAM model.

Takagi et al.'s $\Theta(N\min(K, \sqrt N))$ bound also applies to
the number of visited nodes and that of redirected secondary edges
of $\DAWG(\mathcal{S})$ for multiple strings in the fully-online setting.
Instead, they showed how to simulate secondary edge traversals of
$\DAWG(\mathcal{S})$ in $O(\log \sigma)$ amortized time each,
using the aforementioned NMA structures.
We remark that their data structure is only an implicit representation
of $\DAWG(\mathcal{S})$ in the sense that
the secondary edges are not explicitly stored.

\section{Simple fully-online constructions of suffix trees and DAWGs
  on the pointer-machine model}
\label{sec:algorithms}

In this section, we present our new algorithms
for fully-online construction of suffix trees and DAWGs
for multiple strings, which work on the pointer-machine model.

\subsection{Right-to-left suffix tree construction}

In this section, we present our new algorithm that constructs
the suffix tree for a fully-online right-to-left string collection.

Consider a collection $\mathcal{T}' = \{T_1, \ldots, T_K\}$ of $K$ strings.
Suppose that we have built $\STree(\mathcal{T}')$
and that for each string $T_i \in \mathcal{T}'$ we know the leaf $\ell_i$ that represents $T_i$.

In our fully-online setting,
any new character from $\Sigma$ can be prepended to any string in
the current string collection $\mathcal{T}$.
Suppose that a new character $a \in \Sigma$ is prepended to
a string $T$ in the collection $\mathcal{T}'$,
and let $\mathcal{T} = (\mathcal{T}' \setminus \{T\}) \cup \{aT\}$
be the collection after the update.
Our task is to update $\STree(\mathcal{T}')$ to $\STree(\mathcal{T})$.

Our approach is to reduce the sub-problem of redirecting
Weiner links to the \emph{ordered split-insert-find problem}
that operates on ordered sets over dynamic universe of elements,
and supports the following operations and queries efficiently:
\begin{itemize}
\item Make-set, which creates a new list that consists only of a single element;
\item Split, which splits a given set into two disjoint sets,
so that one set contains only smaller elements than the other set;
\item Insert, which inserts a new single element to a given set;
\item Find, which answers the name of the set that a given element belongs to.
\end{itemize}

Recall our description of Weiner's algorithm in Section~\ref{sec:Weiner_Blumer}
and see Figure~\ref{fig:Weiner_case_2}.
Consider the set of in-coming Weiner links of node $u$ before updates
(the left diagram of Figure~\ref{fig:Weiner_case_2}),
and assume that these Weiner links are sorted by the length of the origin nodes.
After arriving the node $v$ in the climbing up process from the leaf for $T$,
we take the Weiner link with character $a$ and arrive at node $u$.
Then we access the set of in-coming Weiner-links of $u$ by a find query.
When we create a new internal node $y$ as the parent of the new leaf for $aT$,
we split this set into two sets, one as the set of in-coming Weiner links of $y$,
and the other as the set of in-coming Weiner links of $u$
(see the right diagram of Figure~\ref{fig:Weiner_case_2}).
This can be maintained by a single call of a split operation.

Now we pay our attention to the copying process of
Weiner links described in Figure~\ref{fig:Weiner_case_2_copy}.
Observe that each newly copied Weiner links can be inserted
by a single find operation and a single insert operation
to the set of in-coming Weiner links of $\W{u}{c}$ for each character $c$
where $\W{u}{c}$ is defined.

Now we prove the next lemma:
\begin{lemma}
  Let $f$ denote the operation and query time of a linear-space algorithm
  for the ordered split-insert-find problem.
  Then, we can build the suffix tree for a fully-online right-to-left
  string collection of total length $N$ in a total of $O(N (f + \log \sigma))$ time
  and $O(N)$ space.
\end{lemma}

\begin{proof}
  The number of split operations is clearly bounded by
  the number of leaves, which is $N$.
  Since the number of Weiner links is at most $3N-4$,
  the number of insert operations is also bounded by $3N-4$.
  The number of find queries is thus bounded by $N + 3N-4 = 4N-4$.
  By using a linear-space split-insert-find data structure,
  we can maintain the set of in-coming Weiner links
  for all nodes in a total of $O(N f)$ time with $O(N)$ space.

  Given a new character $a$ to prepend to a string $T$,
  we climb up the path from the leaf for $T$ and find the deepest ancestor $v$
  of the leaf for which $\W{a}{v}$ is defined.
  This can be checked in $O(\log \sigma)$ time at each visited node,
  by using a balanced search tree.
  Since we do not climb up the nodes $z$ (see Figure~\ref{fig:Weiner_case_2})
  for which the soft Weiner links with $a$ are redirected, 
  we can use the same analysis (a) as in the case of a single string.
  This results in that the number of visited nodes in our algorithm is $O(N)$.
  Hence we use $O(N \log \sigma)$ total time for finding
  the deepest node which has a Weiner link for the prepended character $a$.

  Overall, our algorithm uses $O(N (f + \log \sigma))$ time and $O(N)$ space.
\end{proof}

Our ordered split-insert-find problem is a special case of
the union-split-find problem on ordered sets,
since each insert operation can be trivially simulated by make-set and union operations.
%P\v{a}tra\c{s}cu and Demine~\cite{PatrascuD04} showed
%a cell-probe $\Omega(\log d)$ lower bound on the dynamic graph connectivity prob%lem,
%where $d$ denote the number of elements.
%In this problem, one can insert and delete elements,
%and check if two elements (nodes) are in the same connected component.
%Since their lower bound holds even for path graphs,
%as long as each set is implemented as an ordered list,
%$\Omega(\log d)$ is also a lower bound for the ordered split-insert-find problem.
Link-cut trees of Sleator and Tarjan~\cite{SleatorT83} for a dynamic forest
support make-tree, link, cut operations and find-root queries
in $O(\log d)$ time each.
Since link-cut trees can be used to path-trees,
make-set, insert, split, and find in the ordered split-insert-find problem
can be supported in $O(\log d)$ time each.
%A more survey on this particular topic can be found in~\cite{Lai2008}.
%
Since link-cut trees work on the pointer machine model,
this leads to a pointer-machine algorithm
for our fully-online right-to-left construction of the suffix tree
for multiple strings with $f = O(\log d)$.
Here, in our context, $d$ denotes the maximum number of in-coming Weiner links
to a node of the suffix tree.
%A simple analysis showsthat the running time
%$O(N (f + \log \sigma)) = O(N(\log d + \log \sigma))$ can also be expressed
%as $O(N \log(d + \sigma))$\footnote{For any $0 < x \leq y$, we have $\log x  + \log  y = \log (xy) \leq \log y^2 = 2 \log y = O(\log y)$ and $\log (x + y) \leq \log (2y) = \log2 + \log y = O(\log y)$, where $\log$ denotes the logarithm with base 2.}.

A potential drawback of using link-cut trees is that
in order to achieve $O(\log d)$-time operations and queries,
link-cut trees use some auxiliary data structures
such as splay trees~\cite{SleatorT85} as its building block.
Yet, in what follows, we describe that
our ordered split-insert-find problem can be solved by
a simpler balanced tree, AVL-trees~\cite{AVL},
retaining $O(N (\log \sigma + \log d))$-time and $O(N)$-space complexities.

\begin{theorem} \label{theo:right-to-left-suffix-tree}
  There is an AVL-tree based pointer-machine algorithm
  that builds the suffix tree for a fully-online right-to-left multiple strings
  of total length $N$ in $O(N (\log \sigma + \log d))$ time with $O(N)$ space,
  where $d$ is the maximum number of in-coming Weiner links to a suffix tree node
  and $\sigma$ is the alphabet size.
\end{theorem}

\begin{proof}
For each node $u$ of the suffix tree $\STree(\mathcal{T}')$ before update,
let $S(u) = \{|x| \mid \W{a}{x} = u\}$ where $a = u[1]$,
namely, $S(u)$ is the set of the string depths of the origin nodes
of the in-coming Weiner links of $u$.
We maintain an AVL tree for $S(u)$ with the node $u$,
so that each in-coming Weiner link for $u$ points
to the corresponding node in the AVL tree for $S(u)$.
The root of the AVL tree is always linked to the suffix tree node $u$,
and each time another node in the AVL tree
becomes the new root as a result of node rotations,
we adjust the link so that it points to $u$ from the new root of the AVL tree.

This way, a find query for a given Weiner link is reduced to
accessing the root of the AVL tree that contains the given Weiner link, 
which can be done in $O(\log S(u)) \subseteq O(\log d)$ time.

Inserting a new element to $S(u)$ can also be done in $O(\log S(u)) \subseteq O(\log d)$ time.

Given an integer $k$, let $S_1$ and $S_2$ denote the subset of $S(u)$
such that any element in $S_1$ is not larger than $k$,
any element in $S_2$ is larger than $k$, and $S_1 \cup S_2 = S(u)$.
It is well known that we can split the AVL tree for $S(u)$ into
two AVL trees for $S_1$ and for $S_2$
in $O(\log S(u)) \subseteq O(\log d)$ time
(c.f.~\cite{Knuth1998}).
In our context, $k$ is the string depth of the deepest node $v$
that is a Weiner link with character $a$ in the upward path from
the leaf for $T$.
This allows us to maintain $S_1 = S(y)$ and $S_2 = S(u)$
in $O(\log d)$ time in the updated suffix tree $\STree(\mathcal{T})$.

When we create the in-coming Weiner links
labeled $a$ to the new leaf $\hat{\ell}$ for $aT$,
we first perform a make-set operation which builds
an AVL tree consisting only of the root.
If we na\"ively insert each in-coming Weiner link to the AVL tree one by one,
then it takes a total of $O(N \log d)$ time.
However, we can actually perform this process in $O(N)$ total time
even on the pointer machine model:
Since we climb up the path from the leaf $\ell$ for $T$,
the in-coming Weiner links are already sorted in decreasing order
of the string depths of the origin nodes.
We create a family of maximal complete binary trees of size $2^h-1$ each,
arranged in decreasing order of $h$.
This can be done as follows:
Initially set $r \leftarrow |S(\hat{\ell})|$.
We then greedily take the largest $h$ such that $2^h-1 \leq r$,
and then update $r \leftarrow r - (2^h-1)$ and search for
the next largest $h$ and so on.
These trees can be easily created in $O(|S(\hat{\ell})|)$ total time
by a simple linear scan over the sorted list of the in-coming Weiner links.
Since the heights $h$ of these complete binary search trees are
monotonically decreasing, and since all of these binary search trees
are AVL trees,
one can merge all of them into a single AVL tree
in time linear in the height of the final AVL tree (c.f.~\cite{Knuth1998}),
which is bounded by $O(h) = O(\log S(\hat{\ell}))$.
Thus, we can construct the initial AVL tree for the in-coming Weiner links
of each new leaf $\hat{\ell}$ in $O(|S(\hat{\ell})|)$ time.
Since the total number of Weiner links is $O(N)$,
we can construct the initial AVL trees for the in-coming Weiner links
of all new leaves in $O(N)$ total time.

Overall, our algorithm works in $O(N (\log \sigma + \log d))$ time with $O(N)$ space.
\end{proof}

\subsection{Left-to-right DAWG construction}

The next theorem immediately follows from Theorem~\ref{theo:right-to-left-suffix-tree}.
\begin{theorem} \label{theo:left-to-right-DAWG}
  There is an AVL-tree based pointer-machine algorithm
  that builds an explicit representation of the DAWG for
  a fully-online left-to-right multiple strings
  of total length $N$ in $O(N (\log \sigma + \log d))$ time with $O(N)$ space,
  where $d$ is the maximum number of in-coming edges of a DAWG node
  and $\sigma$ is the alphabet size.
  This representation of the DAWG allows each edge traversal in
  $O(\log \sigma + \log d)$ time.
\end{theorem}

\begin{proof}
  The correctness and the complexity of construction are immediate from Theorem~\ref{theo:right-to-left-suffix-tree}.

  Given a character $a$ and a node $v$ in the DAWG,
  we first find the out-going edge of $v$ labeled $a$ in $O(\log \sigma)$ time.
  If it does not exist, we terminate.
  Otherwise, we take this $a$-edge and arrive at the corresponding node in the
  AVL tree for the destination node $u$ for this $a$-edge.
  We then perform a find query on the AVL tree and obtain $u$ in $O(\log d)$ time.
\end{proof}

We emphasize that Theorem~\ref{theo:left-to-right-DAWG}
gives the \emph{first} non-trivial algorithm that builds
an explicit representation of the DAWG for fully-online multiple strings.
Recall that a direct application of Blumer et al.'s algorithm
to the case of fully-online $K$ multiple strings
requires to visit $\Theta(N\min(K, \sqrt N))$ nodes in the DAWG,
which leads to $O(N\min(K, \sqrt N) \log \sigma) = O(N^{1.5} \log \sigma)$-time construction for $K = \Theta(\sqrt{N})$.

It should be noted that after all the $N$ characters have been processed,
it is easy to modify, in $O(N)$ time in an offline manner,
this representation of the DAWG so that
each edge traversal takes $O(\log \sigma)$ time.

\subsection{On optimality of our algorithms}

It is known that sorting a length-$N$ sequence of $\sigma$ distinct characters
is an obvious lower bound
for building the suffix tree~\cite{Farach-ColtonFM00} or alternatively the DAWG.
This is because, when we build the suffix tree or the DAWG
where the out-going edges of each node are sorted in the lexicographical order,
then we can obtain a sorted list of characters at their root.
Thus, $\Omega(N \log \sigma)$ is a comparison-based model lower bound for building
the suffix tree or the DAWG.
Since Takagi et al.'s $O(N \log \sigma)$-time
algorithm~\cite{TakagiIABH20} works only on the word RAM model,
in which faster integer sorting algorithms exist,
it would be interesting to explore some cases where
our $O(N (\log \sigma + \log d))$-time algorithms for a weaker model of computation
can perform in optimal $O(N \log \sigma)$ time.

It is clear that the maximum number $d$ of in-coming Weiner links
to a node is bounded by the total length $N$ of the strings.
Hence, in case of integer alphabets of size $\sigma = N^{O(1)}$,
our algorithms run in optimal $O(N \log \sigma) = O(N \log N)$ time.

For the case of smaller alphabet size $\sigma = \polylog(N)$,
the next lemma can be useful:

\begin{lemma}
  The maximum number $d$ of in-coming Weiner links
  is less than the height of the suffix tree.
\end{lemma}

\begin{proof}
  For any node $u$ in the suffix tree,
  all in-coming Weiner links to $u$ is labeled by the same character
  $a$, which is the first character of the substring represented by $u$.
  Therefore, all in-coming Weiner links to $u$ are from the nodes
  in the path between the root and the node $u[2..|u|]$.
\end{proof}

We note that the height of the suffix tree for multiple strings
is bounded by the length of the longest string in the collection.
In many applications such as time series from sensor data,
it would be natural to assume that
all the $K$ strings in the collection have similar lengths.
Hence, when the collection consists of $K = N / \polylog(N)$ strings of length $\polylog(N)$ each, we have $d = \polylog(N)$.
In such cases, our algorithms run in optimal $O(N \log \sigma) = O(N \log \log N)$ time.

\begin{figure}[tbh]
  \centerline{
    \includegraphics[scale=0.4]{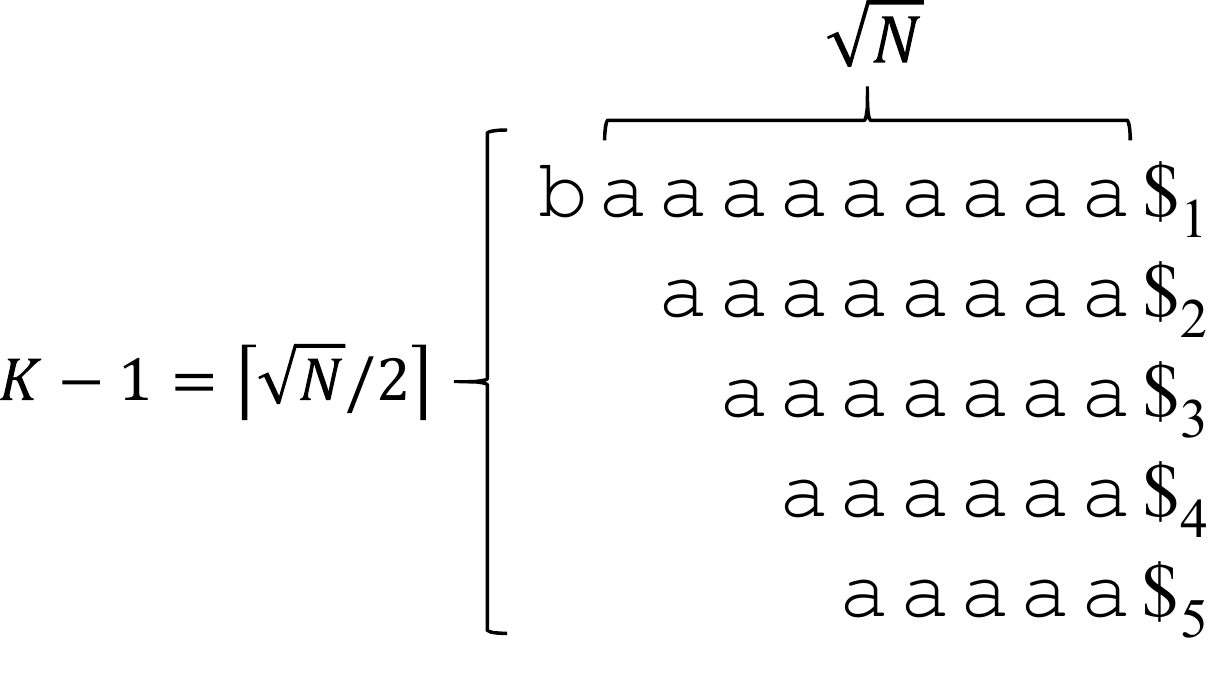}
    \hspace*{2cm}
    \includegraphics[scale=0.45]{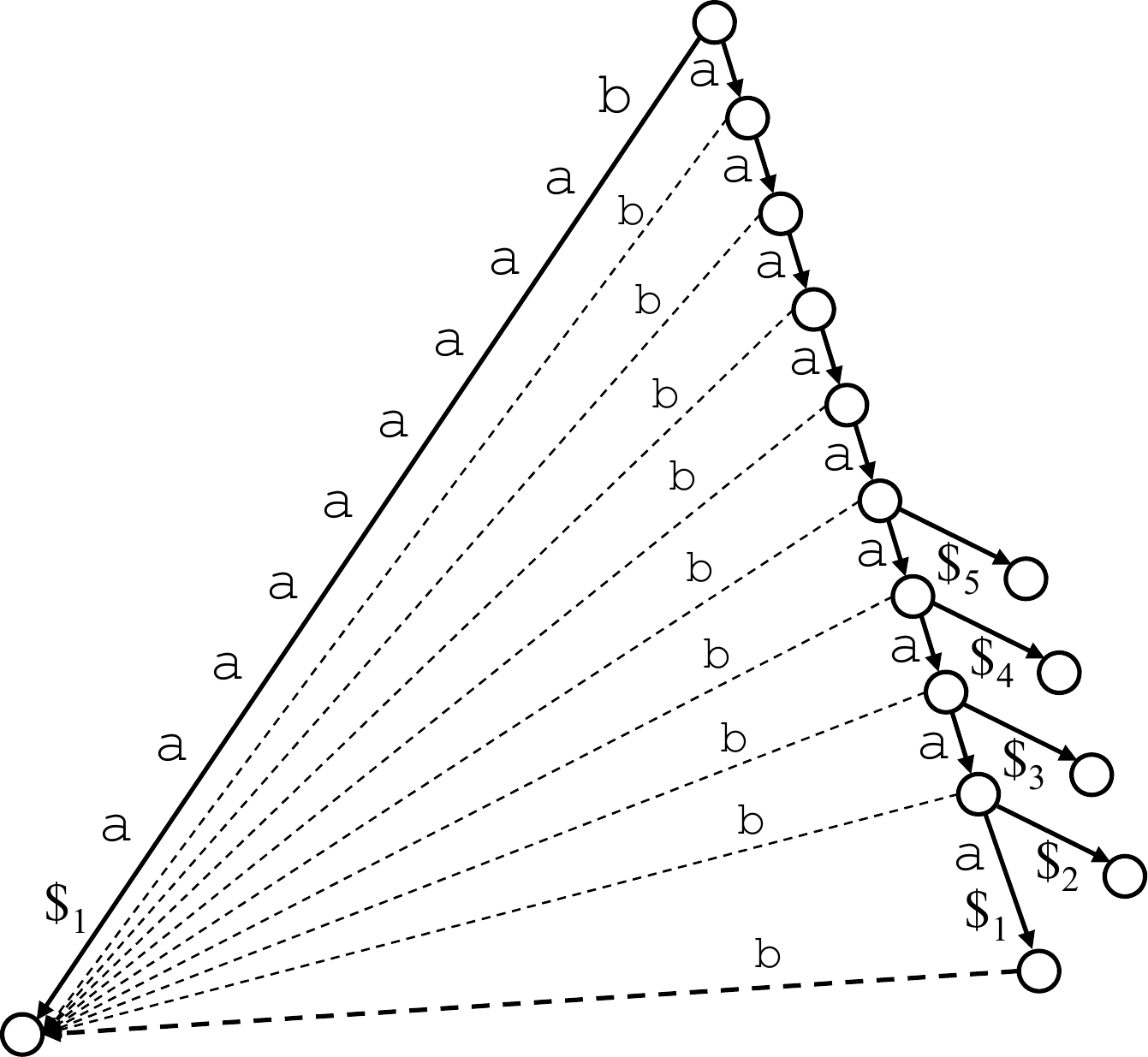}
  }
    \caption{Left: The $K-1 = \lceil \sqrt{N}/2 \rceil$ strings where character $\mathtt{b}$ has been prepended only to the first string $T_1$. Right: The corresponding part of the suffix tree. Dashed arrows represent Weiner links with character $\mathtt{b}$.}
  \label{fig:lowerbound1}
\end{figure}

\begin{figure}[tbh]
  \centerline{
    \includegraphics[scale=0.4]{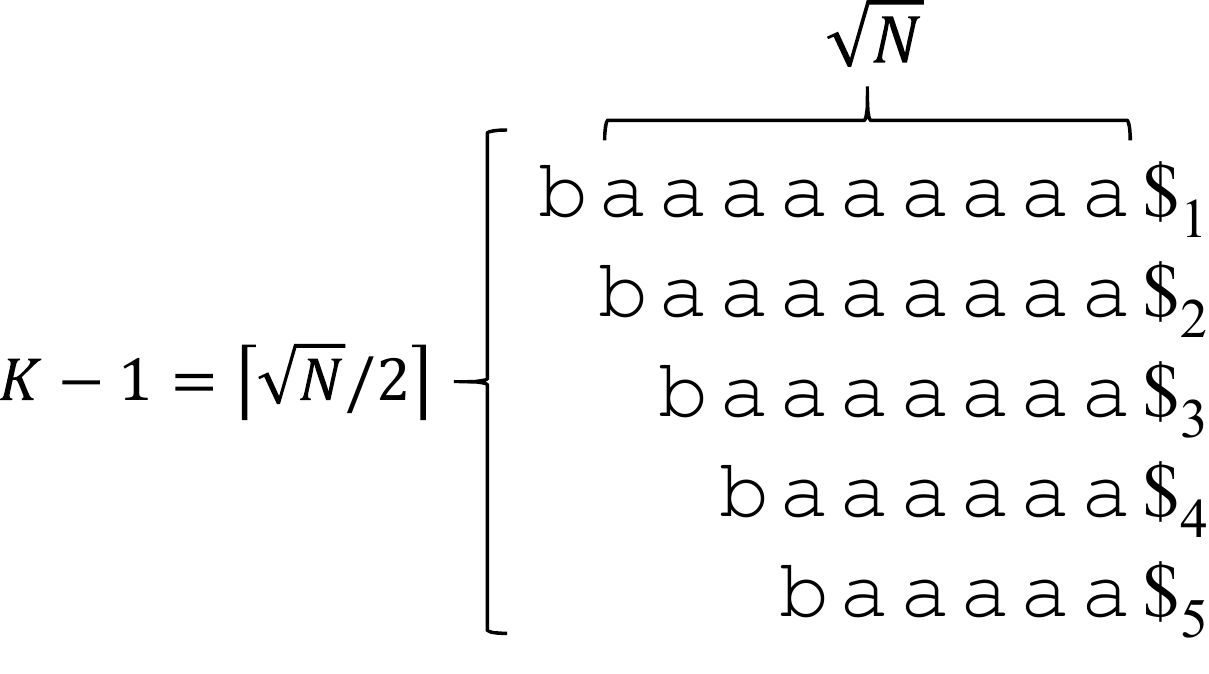}
    \hspace*{2cm}
    \includegraphics[scale=0.45]{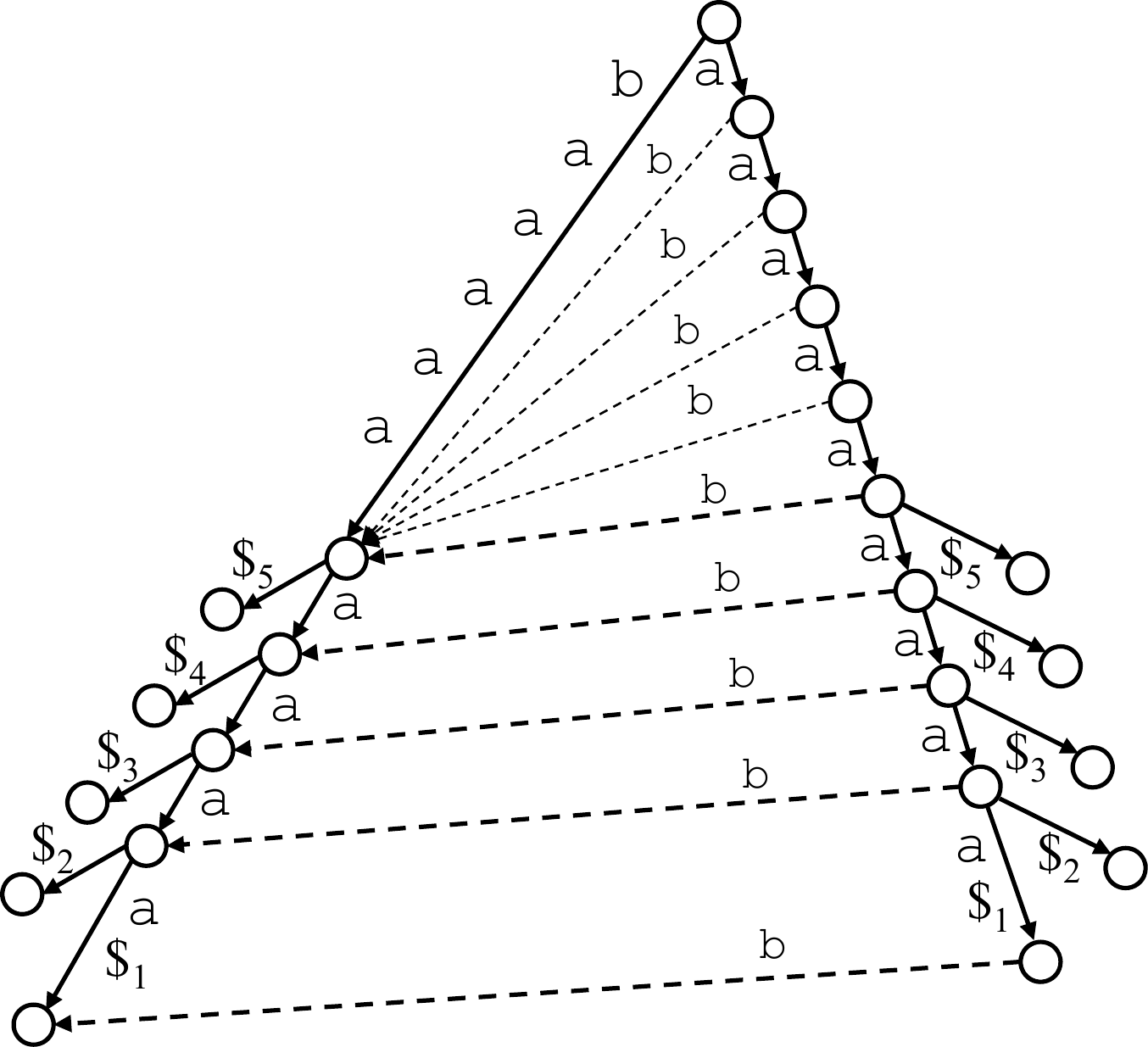}
  }
    \caption{Left: The $K-1 = \lceil \sqrt{N}/2 \rceil$ strings where character $\mathtt{b}$ has been prepended to all of them. Right: The corresponding part of the suffix tree after the updates. Each time a new leaf is created, $\Theta(\sqrt{N})$ in-coming Weiner links were involved in a split operation on the AVL tree and it takes $O(\log N)$ time.}
  \label{fig:lowerbound2}
\end{figure}

The next lemma shows some instance over a binary alphabet of size $\sigma = 2$,
which requires a certain amount of work for the splitting process.

\begin{lemma}
  There exist a set of fully-online multiple strings
  over a binary alphabet
  such that the node split procedure of our algorithms
  takes $O(\sqrt{N} \log N)$ time.
\end{lemma}

\begin{proof}
  Let $K = 1+\lceil \sqrt{N}/2 \rceil$.

  For the time being, we assume that each string $T_i$ is terminated with
  a unique symbol $\$_i$.
  Consider a subset $\{T_1, \ldots, T_{K-1}\}$ of $K-1 = \lceil \sqrt{N}/2 \rceil$
  strings such that for each $1 \leq i \leq K-1$, 
  $T_i = \mathtt{a}^{\sqrt{N}-i+1} \$_i$.
  We then prepend the other character $\mathtt{b}$ from the binary alphabet
  $\{\mathtt{a}, \mathtt{b}\}$ to each $T_i$
  in increasing order of $i = 1, \ldots, K-1$.
  For $i = 1$, $\sqrt{N}$ Weiner links to the new leaf for $\mathtt{b}T_1 = \mathtt{b} \mathtt{a}^{\sqrt{N}} \$_1$, each labeled $\mathtt{b}$, are created.
  See Figure~\ref{fig:lowerbound1} for illustration of this step.

  Then, for each $i = 2, \ldots, K-1$,
  inserting a new leaf for $\mathtt{b}T_i$ requires an insertion of a new internal node
  as the parent of the new leaf.
  This splits the set of in-coming Weiner links into two sets:
  one is a singleton consisting of the Winer link from
  node $\mathtt{a}^{\sqrt{N}-i+1}$, and the other consists of the Weiner links from the shallower nodes.
  Each of these $K-2$ split operations can be done by a simple deletion operation on the corresponding AVL tree,
  using $O(\log \sqrt{N}) = O(\log N)$ time each.
  See Figure~\ref{fig:lowerbound2} for illustration. 
  
  Observe also that the same analysis holds even if we remove
  the terminal symbol $\$_i$ from each string $T_i$
  (in this case, there is a non-branching internal node for each $T_i$
   and we start the climbing up process from this internal node).

  The total length of these $K-1$ strings is approximately $3N/8$.
  We can arbitrarily choose the last string $T_K$ of length approximately $5N/8$
  so that it does not affect the above
  split operations (e.g., a unary string $a^{5N/8}$ or $b^{5N/8}$ would suffice).

  Thus, there exists an instance over a binary alphabet
  for which the node split operations require 
  $O(\sqrt{N} \log N)$ total time.
\end{proof}

Since $\sqrt{N} \log N = o(N)$,
the $\sqrt{N} \log N$ term is always dominated by the $N \log \sigma$ term.
It is left open whether there exists a set of strings
with $\Theta(N)$ character additions,
each of which requires splitting a set that involves
$N^{O(1)}$ in-coming Weiner links.
If such an instance exists,
then our algorithm must take $\Theta(N \log N)$ time in the worst case.

\section{Conclusions and future work}

In this paper we considered the problem of maintaining
the suffix tree and the DAWG indexing structures
for a collection of multiple strings that are updated in a fully-online manner,
where a new character can be added to the left end
or the right end of any string in the collection, respectively.
Our contributions are simple pointer-machine algorithms
that work in $O(N (\log \sigma + \log d))$ time and $O(N)$ space,
where $N$ is the total length of the strings,
$\sigma$ is the alphabet size,
and $d$ is the maximum number of in-coming Weiner links
of a node in the suffix tree.
The key idea was to reduce the sub-problem of re-directing
in-coming Weiner links to the ordered split-insert-find problem,
which we solved in $O(\log d)$ time by AVL trees.
We also discussed the cases where our $O(N (\log \sigma + \log d))$-time solution
is optimal.

A major open question regarding the proposed algorithms
is whether there exists an instance over a small alphabet
which contains $\Theta(N)$ positions
each of which requires $\Theta(\log N)$ time for the split operation,
or requires  $\Theta(N)$ insertions each taking $\Theta(\log N)$ time.
If such instances exist, then the running time of our algorithms
may be worse than the optimal $O(N \log \sigma)$ for small $\sigma$.
So far, we have only found an instance with $\sigma = 2$ that takes
sub-linear $O(\sqrt{N} \log N)$ total time for split operations.

\section*{Acknowledgements}
This work is supported by JST PRESTO Grant Number JPMJPR1922.

\bibliographystyle{abbrv}
\bibliography{ref}

\end{document}